\documentclass[letterpaper,10pt,conference]{IEEEtran}
\usepackage{amsthm,amsfonts,amsmath,amssymb}
\usepackage{algorithm}
\usepackage{algpseudocode}
\usepackage{graphicx}
\usepackage{color}
\usepackage{stmaryrd}
\usepackage{multirow}
\usepackage{longtable}
\usepackage{enumerate}
\usepackage{lscape}
\usepackage{verbatim}
\usepackage{tikz}
\usetikzlibrary{matrix,shapes,shapes.multipart,arrows,positioning,chains,decorations,calc,automata,fit}
\usepackage[top=54pt,right=54pt,left=54pt,bottom=54pt,letterpaper]{geometry}

\DeclareMathOperator*{\argmin}{arg\,min}
\newcommand{\pacc}{p_a}
\newcommand{\pno}{p_n}

\newcommand{\blda}{\mathbf{a}}
\newcommand{\bldb}{\mathbf{b}}

\newcommand{\nn}{\mathbb{N}}
\newcommand{\NN}{\mathbb{N}}
\newcommand{\FF}{\mathbb{F}}

\newcommand{\FFF}{\mathcal{F}}
\newcommand{\cP}{{\mathcal{P}}}
\newcommand{\cR}{{\mathcal{R}}}

\newcommand{\cH}{{\mathcal{H}}}

\DeclareMathOperator*{\disj}{\mbox{\sc Disj}}
\DeclareMathOperator*{\comm}{\mbox{\sc Comm}}

\newcommand{\regzero}{{0}}

\newtheorem{theorem}{\indent Theorem}[section]
\newtheorem{lemma}[theorem]{\indent Lemma}

\newtheorem{example}{\indent Example}[section]
\newtheorem{definition}[theorem]{\indent Definition}

\title{\vspace{18pt}Two-Party Function Computation \\on the Reconciled Data}
\author{{\bf Ivo Kubjas and Vitaly Skachek} \\
Institute of Computer Science \\
University of Tartu, Estonia \\
\texttt{\{ivokub, vitaly.skachek\}@ut.ee}
}
\begin{document}

\maketitle
\footnotetext[1]{This work is supported in part by the grant EMP133 from the
Norwegian-Estonian Research Cooperation Programme and by the grants PUT405 and
IUT2-1 from the Estonian Research Council.}

\begin{abstract}

In this paper, we initiate a study of a new problem termed \emph{function
computation on the reconciled data}, which generalizes a set reconciliation
problem in the literature.  Assume a distributed data storage system with two
users $A$ and $B$. The users possess a collection of binary vectors $S_{A}$ and
$S_{B}$, respectively. They are interested in computing a function $\phi$ of the
reconciled data $S_{A} \cup S_{B}$.

It is shown that any deterministic protocol, which computes a \emph{sum} and a
\emph{product} of reconciled sets of binary vectors represented as nonnegative integers, has to communicate at
least $2^n + n - 1$ and $2^n + n - 2$ bits in the worst-case scenario,
respectively, where $n$ is the length of the binary vectors.  Connections to
other problems in computer science, such as \emph{set disjointness} and
\emph{finding the intersection}, are established, yielding a variety of
additional upper and lower bounds on the communication complexity.  A protocol
for computation of a sum function, which is based on use of a family of hash
functions, is presented, and its characteristics are analyzed.
\end{abstract}

\section{Introduction}

The problem of data synchronization arises in many applications
in distributed data storage systems and data networks.
For instance, consider a number of users that concurrently access and update
a jointly used distributively stored large database. When one of the users
makes an update in the data stored locally, the other users are not immediately aware
of the change, and thus an efficient method for synchronization of the data
is required. This practical problem arises in many systems
that store big amounts of data, including those employed by companies such
as Dropbox, Google, Amazon, and others.

The problem of data synchronization was studied in the literature over the recent years.
A variation of this problem termed \emph{two-party set reconciliation}
considers a scenario, where two users communicate via a direct
bi-directional noiseless channel. The users, $A$ and $B$, possess respective
sets $S_{A}$ and $S_{B}$ of binary vectors.  The users execute a communications
protocol by sending binary messages to each other.  At the end of the protocol,
each of the users knows $S_A \cup S_B$. Set reconciliation problem was first studied
in~\cite{Minsky}. Some of the recent works that investigate this problem
are~\cite{Eppstein, Goodrich, Ivo-thesis, Biff-codes, Skachek-Rabbat}.
A number of protocols for set reconciliation were proposed, and their theoretical
performance was analyzed. All aforementioned protocols communicate amount of data,
which is asymptotically optimal.

In practical data storage systems, sometimes only a function of the stored data can be requested by
some user, and not the data itself. It can be more efficient to compute a function by a group of servers, rather that
to provide the full data required for such a computation by the user (see, for instance, Example~\ref{ex:max} below).
Therefore, it is an important question how to compute various functions of the data
distributed among a number of servers.

The domain of distributed function computation is a mature area, which has been very extensively studied
both in computer science and information theory communities. The reader can refer,
for example, to~\cite{Kushilevitz-Nisan},~\cite{Lovasz},~\cite{Orlitsky},~\cite{Orlitsky2},~\cite{Yao1979},
and many others. In a standard model, a number of users want to compute jointly a function of the
data that they possess. This needs to be achieved by communicating the smallest possible
number of bits. This class of problems is very broad, and it covers settings with various types of functions, two versus many users, deterministic and randomized protocols, with or without privacy requirements, etc.

Motivated by the above challenges, in this work, we propose a new problem,
which we term \emph{function computation on the reconciled data}.
To the best of our knowledge, this problem was not studied in the literature yet.
In this problem, the users compute a function of their reconciled data.
It is obvious that this problem can be solved by reconciling
the data first, and then by computing the function of this data by the users.
However, as we demonstrate in the sequel, this approach is not always optimal in
terms of a number of communicated bits.

This paper is structured as follows. In Section~\ref{sec:settings},
the problem of function computation on the reconciled data is introduced.
In Section~\ref{sec:reconciliation}, known methods for set reconciliation are surveyed.
It is shown that using reconciliation as a subroutine does not necessarily yield an optimal solution.
A number of bounds on the communication complexity of sum computation on the reconciled data
are obtained in Section~\ref{section:lower_bounds}. Connections to some known problems in computer science
are established in Section~\ref{sec:connections}. A protocol for
computation of sum using universal hash functions and its analysis are presented in Section~\ref{sec:hash_functions}.
The results are summarized in Section~\ref{sec:summary}.

\section{Problem settings}
\label{sec:settings}

Let $\FF = \{0,1\}$ be a binary field. Denote by $\FF^n$ the vector space of
dimension $n$ over $\FF$. By slightly abusing the notation, sometimes we treat
$\FF^n$ as a set of all vectors of length $n$ over $\FF$, or, as a set of nonnegative integers in their $n$-bit long binary representation. Let the set of all
subsets of $\FF^n$ be $\cP(\FF^n)$.  We denote $[\ell] \triangleq \{1, 2,
\cdots, \ell\}$.

Consider two users, $A$ and $B$, possessing sets $S_{A}, S_{B} \subseteq \FF^n$,
respectively. We denote the intersection of these two sets as $S_0 = S_{A} \cap
S_{B}$. The sizes of these sets are given as $m_0 = |S_0|$, $m_{A} = |S_{A}|$
and $m_{B} = |S_{B}|$. Additionally, it is assumed that $\max\{m_{A}, m_{B}\}
\le \kappa$.  Denote the sizes of the set differences as $d_{A} = |S_{A}
\setminus S_0|$, $d_{B} = |S_{B} \setminus S_0|$ and $d = d_A + d_B$. We assume
hereby that $A$ knows the values of $d_{A}$ and $m_0$, and that $B$ knows the
values of $d_{B}$ and $m_0$.

The users $A$ and $B$ want to compute cooperatively a function $f : \cP(\FF^n)
\times \cP(\FF^n) \rightarrow V$, where $V$ is the range of $f$.  The functions
that we consider in this work are all defined over the reconciled data, namely
they have the form $f(S_A, S_B) = \phi(S_{A} \cup S_{B})$, where $S_{A} \cup
S_{B}$ is a standard set-theoretic union of the two sets, and $\phi : \cP(\FF^n)
\rightarrow V$.  In order to do so, $A$ and $B$ jointly execute a communications
protocol, according to which they send binary messages to each other.
Specifically, the protocol $F$ consists of the messages
\begin{align*}
    & M_1 = (w_{1,1}, w_{1,2}, \ldots, w_{1,p_{1}}) \in \FF^{p_1},\\
    & M_2 = (w_{2,1}, w_{2,2}, \ldots, w_{2,p_{2}}) \in \FF^{p_2},\\
    & \vdots\\
    & M_r = (w_{r,1}, w_{r,2}, \ldots, w_{r,p_{r}}) \in \FF^{p_r} ,
\end{align*}
which are sent alternately between $A$ and $B$.  After the message $M_r$ is
sent, both users can compute the value of $f(S_{A},S_{B})$.  The number of
messages $r$ is called the number of rounds of the protocol.

Communication complexity $\comm\left(F\right)$ of the protocol $F$ is defined as
the minimum total number of bits $\sum_{i=1}^r p_i$ that are sent between the
users in the worst-case scenario for all $S_{A}, S_{B} \in \cP(\FF^n)$.

There are different models of how the protocols use randomness. In
\emph{deterministic} protocol, we assume that all computations and messages sent
by the users are deterministic, and they are uniquely determined by the sets
$S_A$ and $S_B$. By following the discussion
in~\cite{DBLP:journals/toc/HastadW07}, we consider several \emph{randomized}
protocol models. In a protocol with \emph{shared randomness}, both users $A$ and
$B$ have access to an infinite sequence of independent unbiased random bits. The
users are expected to compute the function correctly with probability close to
1. By contrast, in a protocol with \emph{private randomness}, each user
possesses its own string of random bits.  Finally, in the ``Las-Vegas''-type
protocol, at the end of the protocol the users always compute the function
correctly, but the number of communicated bits is a random variable, and the
complexity is measured as the expected number of the communicated bits.

\section{Connection to set reconciliation}
\label{sec:reconciliation}

The set reconciliation problem can be viewed as a function computation problem
on the reconciled data, where the function $\phi$ is an identity, namely,
$f(S_A, S_B) = S_A \cup S_B$. A number of protocols were proposed in the
literature for efficient distributed set reconciliation with two users.
In~\cite{Minsky}, interpolation of characteristic polynomials over a Galois
field is used.  The proposed deterministic protocol assumes the knowledge of
approximate values of $d_A$ and $d_B$, and it achieves $\comm(F) = O(d n)$,
which is asymptotically communication-optimal. In particular, when $d$ is small
compared to $n$, that protocol clearly outperforms a naive reconciliation
scheme, where the users simply exchange their data.

Another randomized protocol, which employs invertible Bloom filters, was
presented in~\cite{Eppstein, Goodrich}. Alternatively, it was proposed to use
so-called $\emph{biff codes}$ for randomized set reconciliation
in~\cite{Biff-codes}. Finally, a randomized protocol that uses techniques akin
to linear network coding were employed in~\cite{Skachek-Rabbat} leading to yet
another reconciliation protocol. The latter method assumes existence of certain
family of pseudo-random hash functions. All mentioned randomized protocols have
asymptotically optimal communication complexity $\comm(F) = O(d n)$.

We note that a problem of computing any function $f$ can be solved by $A$ and
$B$ by reconciling their data first, and then by computing $f$ by each user
separately (or by one of the users). By using this method, the communication
complexity is determined by the complexity of the underlying set reconciliation
protocol. For example, for each of the aforementioned protocols, $\comm(F) = O(d
n)$. Sometimes, an improvement in communication complexity can be obtained by
using one-directional reconciliation, namely, when the data is reconciled by
only one user, and then the function value is sent back to the other user.
However, if $d_A \approx d_B$, this approach does not lead to asymptotic
improvement.

As the following example illustrates, some functions can be computed by a
deterministic protocol with much smaller communication complexity.

\begin{example}
\label{ex:max}
Assume that $A$ and $B$ are interested in computing $f(S_A, S_B) = \max\{S_A
\cup S_B \}$, where all entries in $S_A \cup S_B$ are viewed as non-negative
integer numbers in their binary representation.  The following protocol requires
only $2n$-bit communication.
\begin{enumerate}
\item
The users $A$ and $B$ compute $x_A = \max\{ S_A \}$ and $x_B = \max\{ S_B \}$,
respectively.
\item
The users $A$ and $B$ exchange the values of $x_A$ and $x_B$.
\item
Each user computes $\max\{ x_A, x_B\}$.
\end{enumerate}
\end{example}

Analogous protocol can be used to compute a number of other idempotent functions
$\phi$, such as \emph{minimum}, bit-wise logical \emph{or} and bit-wise logical
\emph{and}. It is an interesting question, however, what is the worst-case
number of communicated bits for computing different functions on the reconciled
data. We partly answer this question for some of the functions in the sequel.

\section{Lower bounds using $f$-monochromatic rectangles}
\label{section:lower_bounds}

\subsection{Sum over integers}
\label{subsection:sum_integers}

In this section, we consider the function $f$ with the integer range, defined as
follows:
\begin{equation}
    f(S_{A}, S_{B}) = \sum\limits_{x \in S_{A} \cup S_{B}} x \; ,
    \label{eq:funct_sum}
\end{equation}
where every string $x \in S_{A} \cup S_{B}$ can be viewed as an integer in its
binary representation.

We introduce the following definition, which is taken from~\cite[Definition
1]{Kushilevitz97}.

\begin{definition}
Let $\eta \in \nn$ and $f \, : \, \FF^\eta \times \FF^\eta \rightarrow V$ be a
function with range $V$.  A rectangle is a subset of $\FF^\eta \times \FF^\eta$
of the form $X_1 \times X_2$, where $X_1, X_2 \subseteq \FF^\eta$.  A rectangle
$X_1 \times X_2$ is called $f$-monochromatic if for every $x \in X_1$ and $y \in
X_2$, the value of $f(x, y)$ is the same.
\end{definition}

% TODO: add comment about the lemma
\begin{lemma} \cite[Proposition 1.13]{Kushilevitz-Nisan}
\label{lemma:rectangle-diagonal}
Let $R \subseteq \FF^\eta \times \FF^\eta$. Then $R$ is a rectangle if and only
if
\begin{align}
    (x_1, y_1) \in R \mbox{ and } (x_2, y_2) \in R \;
    \Longrightarrow \; (x_1, y_2) \in R \; .
\end{align}
\end{lemma}

\begin{definition} \cite{Kushilevitz97}
Let $f \; : \; \FF^\eta \times \FF^\eta \rightarrow V$ be a function.  Denote by
$\cR(f)$ the minimum number of $f$-monochromatic rectangles that partition the
space of $\FF^\eta \times \FF^\eta$.
\end{definition}

We use the following lemma, which is stated in~\cite[Lemma 2]{Kushilevitz97}.
It allows to reformulate the problem of lower-bounding communication complexity
as a problem in combinatorics.

\begin{lemma}
\label{lemma:rectangles}
Let $f \; : \; \FF^\eta \times \FF^\eta \rightarrow V$ be a function, which is
computed using protocol $F$. Then,
\[
    \comm(F) \ge \log_2 (\cR(f)) \; .
\]
\end{lemma}
The proof of the lemma is given in~\cite{Kushilevitz-Nisan}.

In order to be able to use Lemma~\ref{lemma:rectangles}, we need to represent
the inputs $S_A$ and $S_B$ as binary vectors. A natural way to do that is by
using binary characteristic vectors $\blda$ and $\bldb$ of length $\eta = 2^n$.

\begin{theorem}
\label{thrm:sum-improved}
The number of bits communicated between $A$ and $B$ in any
deterministic protocol $F$ that computes the function $f$ defined
in~(\ref{eq:funct_sum}) is at least
\[
    \comm(F) \ge 2^n + n - 1 \; .
\]
\end{theorem}
\begin{proof}
The proof is done by estimating the number of $f$-monochromatic rectangles,
where $f$ is given by~(\ref{eq:funct_sum}).

Denote $\Phi \triangleq \FF^n \setminus \{0\}$, where the elements of $\Phi$ can
be viewed as integers in $[2^n-1]$.  We use the following set of pairs of
subsets
\[
    \FFF_0 = \left\{ (Y, \Phi {\setminus} Y): Y \subseteq \Phi \right\}
    \triangleq \{(Y_i, Y'_i): i \in [2^{2^n-1}]\} \; .
\]

Then, for every $(Y_i, Y'_i) \in \FFF_0$, we have
\begin{align*}
    f(Y_i, Y'_i) = \sum_{i=1}^{2^n-1} i = 2^{n-1} (2^n - 1) \; .
\end{align*}

On the other hand, take $i, j \in [2^{2^n-2}]$, such that $i \neq j$.  We have
two cases:
\begin{itemize}
\item
If $Y_i \cup Y'_j \neq \Phi$, then there exists $x \in \Phi$, such that $x
\notin Y_i \cup Y'_j$. In that case, clearly,
\[
f(Y_i, Y'_j) < 2^{n-1} (2^n - 1) \; .
\]
\item
If $Y_i \cup Y'_j = \Phi$, since $S_i \neq S_j$, there exists $x \in Y_i \cap
Y'_j$. Thus, $x \not\in Y'_i \cup Y_j$, and therefore
\[
f(Y_j, Y'_i) < 2^{n-1} (2^n - 1) \; .
\]
\end{itemize}

Therefore, due to Lemma~\ref{lemma:rectangle-diagonal}, there are at least $2^{2^n-1}$ different $f$-monochromatic rectangles
consisting of the elements of $\FFF_0$.

Additionally, for any $\ell \in [2^n - 1]$, denote $\Phi_\ell \triangleq \FF^n
\setminus \{0,\ell\}$.  We use the following pairs
\[
    \FFF_\ell = \{ (Z, \Phi_\ell {\setminus} Z): Z \subseteq \Phi_\ell \}
		\triangleq \{(Z_i, Z'_i): i \in [2^{2^n-2}]\} \; .
\]

Then, for every $(Z_i, Z'_i) \in \FFF_\ell$, we have
\begin{align*}
    f(Z_i, Z'_i) = \sum_{i=1}^{2^n-1} i - \ell = 2^{n-1} (2^n - 1) - \ell \; .
\end{align*}

On the other hand, take $i, j \in [2^{2^n-1}]$, such that $i \neq j$.  Similarly
to the previous case, it can be shown that either
\[
f(Z_j, Z'_i) < 2^{n-1} (2^n - 1) - \ell \; \mbox{ or } \; f(Z_i, Z'_j) < 2^{n-1} (2^n - 1) - \ell \; .
\]

Therefore, due to Lemma~\ref{lemma:rectangle-diagonal}, there are at least $2^{2^n-2}$ different $f$-monochromatic rectangles
consisting of the elements of $\FFF_\ell$.  Since $\ell$ can be chosen in
$2^n-1$ ways, we conclude that the number of different $f$-monochromatic
rectangles is at least
\begin{eqnarray*}
\cR(f) & \ge & 2^{2^n-1} + (2^n-1) \cdot (2^{2^n-2}) \\
& = & (2^{2^n-2}) \cdot \left( 2^n+1 \right) \\
& > & 2^{2^n + n -2} \; .
\end{eqnarray*}
Finally, by applying Lemma~\ref{lemma:rectangles}, and by rounding the result up
to the next bit, we obtain that
$\comm(F) \ge 2^n + n - 1$.
\end{proof}

\begin{example}

\begin{figure}
    \newcommand{\FIGPURP}{white}
    \newcommand{\FIGMAINPURP}{purple}
    \newcommand{\FIGBLUE}{white}
    \newcommand{\FIGMAINBLUE}{blue}
    \newcommand{\FIGGREEN}{white}
    \newcommand{\FIGYEL}{white}
    \newcommand{\FIGBLACK}{black}
    \newcommand{\FIGSIDEBLACK}{white}
    \newcommand{\FIGBROWN}{white}
    \newcommand{\FIGMAINBROWN}{brown}
    \begin{tikzpicture}[ampersand replacement=\&]
        \matrix [matrix of math nodes] (m)
        {%
            \& \emptyset \& \{1\} \& \{2\} \& \{3\} \& \{1,2\} \& \{1,3\} \&
            \{2,3\} \& \{1,2,3\} \\
            \{1,2,3\} \&
            |[fill=\FIGBLACK!20]| 6 \& |[fill=\FIGSIDEBLACK!20]| 6 \&
            |[fill=\FIGSIDEBLACK!20]| 6 \& |[fill=\FIGSIDEBLACK!20]| 6 \&
            |[fill=\FIGSIDEBLACK!20]| 6 \& |[fill=\FIGSIDEBLACK!20]| 6 \&
            |[fill=\FIGSIDEBLACK!20]| 6 \& |[fill=\FIGSIDEBLACK!20]| 6 \\
            \{2,3\} \&
            |[fill=\FIGMAINPURP!20]| 5 \& |[fill=\FIGBLACK!20]| 6 \&
            |[fill=\FIGPURP!20]| 5 \& |[fill=\FIGPURP!20]| 5 \&
            |[fill=\FIGSIDEBLACK!20]| 6 \& |[fill=\FIGSIDEBLACK!20]| 6 \&
            |[fill=\FIGPURP!20]| 5 \& |[fill=\FIGSIDEBLACK!20]| 6 \\
            \{1,3\} \&
            |[fill=\FIGMAINBROWN!20]| 4 \& |[fill=\FIGBROWN!20]| 4 \&
            |[fill=\FIGBLACK!20]| 6 \& |[fill=\FIGBROWN!20]| 4 \&
            |[fill=\FIGSIDEBLACK!20]| 6 \& |[fill=\FIGBROWN!20]| 4 \&
            |[fill=\FIGSIDEBLACK!20]| 6 \& |[fill=\FIGSIDEBLACK!20]| 6 \\
            \{1,2\} \&
            |[fill=\FIGMAINBLUE!20]| 3 \& |[fill=\FIGBLUE!20]| 3 \&
            |[fill=\FIGBLUE!20]| 3 \& |[fill=\FIGBLACK!20]| 6 \&
            |[fill=\FIGBLUE!20]| 3 \& |[fill=\FIGSIDEBLACK!20]| 6 \&
            |[fill=\FIGSIDEBLACK!20]| 6 \& |[fill=\FIGSIDEBLACK!20]| 6 \\
            \{3\} \&
            |[fill=\FIGBLUE!20]| 3 \& |[fill=\FIGMAINBROWN!20]| 4 \&
            |[fill=\FIGMAINPURP!20]| 5 \& |[fill=\FIGBLUE!20]| 3 \&
            |[fill=\FIGBLACK!20]| 6 \& |[fill=\FIGBROWN!20]| 4 \&
            |[fill=\FIGPURP!20]| 5 \& |[fill=\FIGSIDEBLACK!20]| 6 \\
            \{2\} \&
            |[fill=\FIGGREEN!20]| 2 \& |[fill=\FIGMAINBLUE!20]| 3 \&
            |[fill=\FIGGREEN!20]| 2 \& |[fill=\FIGMAINPURP!20]| 5 \&
            |[fill=\FIGBLUE!20]| 3 \& |[fill=\FIGBLACK!20]| 6 \&
            |[fill=\FIGPURP!20]| 5 \& |[fill=\FIGSIDEBLACK!20]| 6 \\
            \{1\} \&
            |[fill=\FIGYEL!20]| 1 \& |[fill=\FIGYEL!20]| 1 \&
            |[fill=\FIGMAINBLUE!20]| 3 \& |[fill=\FIGMAINBROWN!20]| 4 \&
            |[fill=\FIGBLUE!20]| 3 \& |[fill=\FIGBROWN!20]| 4 \&
            |[fill=\FIGBLACK!20]| 6 \& |[fill=\FIGSIDEBLACK!20]| 6 \\
            \emptyset \&
             0 \&
            |[fill=\FIGYEL!20]| 1 \&
            |[fill=\FIGGREEN!20]| 2 \& |[fill=\FIGBLUE!20]| 3 \&
            |[fill=\FIGMAINBLUE!20]| 3 \& |[fill=\FIGMAINBROWN!20]| 4 \&
            |[fill=\FIGMAINPURP!20]| 5 \& |[fill=\FIGBLACK!20]| 6 \\
        };
    \end{tikzpicture}
    \caption{Example of $f$-monochromatic rectangles in the proof of
        Theorem~\ref{thrm:sum-improved} for $n=2$}
    \label{fig:sumproof}
\end{figure}

In Figure~\ref{fig:sumproof}, we show $f$-monochromatic rectangles whose
existence is proved in Theorem~\ref{thrm:sum-improved}. Four sets of
$f$-monochromatic rectangles,   $\FFF_0$, $\FFF_1$, $\FFF_2$ and $\FFF_3$, are
shown in four different colors.  Each set contains a number of a single-entry
$f$-monochromatic rectangles.

We see that the total number of monochromatic rectangles is at least
\begin{eqnarray*}
    \cR(f) & \geq & |\FFF_0| + |\FFF_1| + |\FFF_2| + |\FFF_3| \\
        & = & 8 + 4 + 4 + 4  \\
        & = & 20 \; .
\end{eqnarray*}

By using Lemma~\ref{lemma:rectangles}, the communication complexity is at least
$\log_2(\cR(f)) = \log_2(20)$ bits. By rounding up to the next integer, we
obtain that $\comm(f) \geq 5$.

We remark that the result can be slightly improved by using the fact that there
are additional rectangles corresponding to the values $0$, $1$ and $2$. However,
that improvement is relatively small, and thus we omit it for the sake of
simplicity.
\end{example}

We also note that there is a trivial deterministic protocol that computes $f$ by
using $2^n + 2n - 2$ bits: first, $A$ sends the  characteristic vector $\blda$
of $S_A$ of length $2^n-1$ (note that zero does not effect the sum) to $B$, then
$B$ computes $f$ and sends the result back to $A$. Since the sum requires $2n-1$
bits to represent, the claimed result follows.

\subsection{Multiplication over integers}
\label{subsection:mult_integers}

As before, let $S_{A}, S_{B} \subseteq \FF^n$. Consider the function $f$ with
the integer range, defined as follows:
\begin{equation}
    f(S_{A}, S_{B}) = \prod\limits_{x \in S_{A} \cup S_{B}} x \; .
    \label{eq:funct_prod}
\end{equation}
The following theorem presents a lower bound on the communication complexity of
a two-party deterministic protocol for computation of this $f$.
\begin{theorem}
The number of bits communicated between $A$ and $B$ in any deterministic
protocol $F$ that computes the function $f$ defined in~(\ref{eq:funct_prod}) is
at least
\[
    \comm(F) \ge 2^n + n -2 \; .
\]
\end{theorem}

\begin{proof}
    The proof is analogous to the proof of Theorem~\ref{thrm:sum-improved}. We
    estimate the number of different $f$-monochromatic rectangles, and then
    apply Lemma~\ref{lemma:rectangles} to obtain a lower bound on the
    communication complexity.

    Denote $\Phi \triangleq \FF^n \setminus \{0,1\}$. At first, we count the
    number of rectangles on the main diagonal. We define:
    \[
        \FFF_0 = \left\{(Y, \Phi {\setminus} Y): Y \subseteq \Phi \right\}
        \triangleq \{(Y_i, Y'_i): i \in [2^{2^n-2}]\}.
    \]

    Then, for every $(Y_i, Y'_i) \in \FFF_0$:
    \[
        f(Y_i, Y'_i) = \prod_{i=2}^{2^n-1} i = (2^n-1)!.
    \]

    Take $i,j \in [2^{2^n-2}]$ such that $i \neq j$. We consider two cases:
    \begin{itemize}
        \item If $Y_i \cup Y'_j \neq \Phi$, then there exists $x \in \Phi$, such
            that $x \not\in Y_i \cup Y'_j$. Then,
            \[
                f(Y_i,Y'_j) < (2^n-1)!.
            \]
        \item
            If $Y_i \cup Y'_j = \Phi$, since $Y_i \neq Y_j$, there exists $x \in
            Y_i \cap Y'_j$, thus $x \not\in Y'_i \cup Y_j$. Then,
            \[
                f(Y_j, Y'_i) < (2^n-1)!.
            \]
    \end{itemize}

    Due to Lemma~\ref{lemma:rectangle-diagonal}, there exist at least $2^{2^n-2}$ different $f$-monochromatic
    rectangles in $\FFF_0$.

    Additional $f$-monochromatic rectangles can be constructed as follows. For
    every $\ell \in \{2, \ldots, 2^n-1\}$, denote
    $\Phi_\ell \triangleq \FF^n \setminus\{0,1,\ell\}$. We define the pairs
    \[
        \FFF_\ell = \left\{(Z, \Phi_\ell {\setminus} Z): Z \subseteq
        \Phi_\ell\right\} \triangleq \{(Z_i, Z'_i): i \in [2^{2^n-3}]\}.
    \]

    Then, for every pair $(Z_i, Z'_i) \in \FFF_\ell$ we have that
    \[
        f(Z_i, Z'_i) = \prod_{\substack{i=2 \\ i\neq\ell}}^{2^n-1} i =
        \frac{(2^n-1)!}{\ell}.
    \]

    Take $i,j \in [2^{2^n-3}]$ such that $i \neq j$. Then, similarly to the
    proof of Theorem~\ref{thrm:sum-improved}, either
    \[
        f(Z_j,Z'_i) < \frac{(2^n-1)!}{\ell}
    \]
    or
    \[
        f(Z_i,Z'_j) < \frac{(2^n-1)!}{\ell} \; .
    \]

    From Lemma~\ref{lemma:rectangle-diagonal}, the set $\FFF_\ell$ contains $2^{2^n-3}$ $f$-monochromatic rectangles.
    We can choose $\ell$ in $2^n-2$ ways, and thus the number of
    $f$-monochromatic rectangles in $\FFF_\ell$, $\ell \neq 0$, is
    \begin{eqnarray}
        (2^n-2) \cdot (2^{2^n-3}) \; . \label{eq:prod_fl}
    \end{eqnarray}
    There is at least one additional $f$-monochromatic rectangle
    corresponding to the value $0$ of the function $f$. By summing things
    up, we obtain that the total number of $f$-monochromatic rectangles is
    at least
    \begin{eqnarray*}
        \cR(f) & \ge & 2^{2^n-2} + (2^n-2) \cdot (2^{2^n-3}) + 1 \\
        & = & 2^{2^n+n-3} + 1.
    \end{eqnarray*}

    Due to Lemma~\ref{lemma:rectangles}, by rounding up to the next
    integer, the communication complexity of a protocol $F$ computing $f$ as
    defined in Equation~\ref{eq:funct_prod} is at least $\comm(F) \geq 2^n+n-2$.
\end{proof}

\section{Connections to Known Problems}
\label{sec:connections}

\subsection{Lower Bounds using Results for Set Disjointness}
\label{sec:reduction-set-disjointness}

Given two sets $S_A, S_B \subseteq \FF^n$, the binary set disjointness function
$\disj{(S_A, S_B)}$ is defined as follows:
\begin{equation}
    \disj{(S_A, S_B)} = \left\{ \begin{array}{cl}
    1 & \mbox{ if } S_A \cap S_B = \varnothing \\
    0 & \mbox{ otherwise }
    \end{array} \right. \; .
\end{equation}

\noindent
{\bf Set disjointness problem:} there are two users $A$ and $B$ that possess the
sets $S_A, S_B \subseteq \FF^n$, respectively. The users want to compute jointly
the function $\disj{(S_A, S_B)}$.

We show a simple reduction from the set disjointness problem to the sum
computation problem.

\noindent
\textbf{Reduction:} assume that $F$ is a protocol for computing $f$
in~(\ref{eq:funct_sum}) by $A$ and $B$. Then, given $S_{A}$ and $S_{B}$, the set
disjointness problem can be solved by $A$ and $B$ as follows.
\begin{enumerate}
\item
The user $A$ sends to $B$ a special bit, indicating if $\regzero \in A$.  If
$\regzero \in A \cap B$, then $B$ announces that $\disj{(S_{A}, S_{B})} = 0$.
Halt.
\item
The users $A$ and $B$ compute $x_A = \sum_{x \in S_{A}} x$ and $x_B = \sum_{x
\in S_{B}} x$, respectively.
\item
The users $A$ and $B$ run the protocol $F$ to find $y = f(S_{A}, S_{B})$.
\item
User $B$ sends $x_B$ to $A$.
\item
If $x_A + x_B = y$, then $A$ concludes that $\disj{(S_{A}, S_{B})} = 1$.
Otherwise, if $x_A + x_B \neq y$, then $\disj{(S_{A}, S_{B})} = 0$.
\end{enumerate}
The correctness of the protocol is straightforward, given that $S_{A} \cap
S_{B} = \varnothing$ if and only if $x_A + x_B = y$ and $\regzero \notin A \cap
B$.

A single bit is sent in Step 1 and $2n-1$ bits are required to represent the
integer value of $x_B$ in Step 4. Thus, the communication complexity of the
proposed protocol for the set disjointness problem is $\comm(F) + 2n$. Then, the
upper bound for set disjointness problem is $\comm(F) + 2n \geq \comm(\disj)$.

There is a variety of known bounds on
communication complexity of the two-party protocols for the set disjointness
problem. For example, for deterministic protocols, there is a lower bound of $2^n
+ 1$ bits \cite{Kushilevitz-Nisan} using fooling sets, and for randomized
protocols the asymptotically tight bound is
$\Theta(2^n)$~\cite{DBLP:journals/jcss/Bar-YossefJKS04,
DBLP:journals/toc/HastadW07, KS, Razborov}. From these bounds, we obtain the
lower bounds $\comm(F) \geq 2^n - 2n+1$ for deterministic and $\comm(F) =
\Omega(2^n)$ for randomized case of function computation problem.

Recall that for the deterministic case, there is an upper bound
of $O(2^n)$ for sum computation problem (see discussion at the end of Section~\ref{subsection:sum_integers}),
which is also an upper bound on complexity of any randomized protocol, thus yielding an asymptotically tight bound of $\Theta(2^n)$ for randomized settings.

\begin{table*}[t]
\begin{center}
    \begin{tabular}{|p{0.15\textwidth}|c|p{0.30\textwidth}|p{0.26\textwidth}|}
        \hline \hline
        Communication Complexity & Protocol Type & Comments & Source
        \\
        \hline
        $\Theta (d \cdot n)$
        &
        Deterministic
        &
        Reconciliation first, difference size is $d$
        &
        Section~\ref{sec:reconciliation} and \cite{Minsky}
        \\
        $\ge 2^n + n - 1$
        &
        Deterministic
        &

        &
        Subsection~\ref{subsection:sum_integers}
        \\
        $\le 2^n + 2n - 2$
        &
        Deterministic
        &

        &
        Subsection~\ref{subsection:sum_integers}
        \\
        $\ge 2^n - 2n + 1$
        &
        Deterministic
        &
        Reduction to set disjointness
        &
        Subsection~\ref{sec:reduction-set-disjointness} and
        \cite{Kushilevitz-Nisan}
        \\
        $\Theta(2^n)$
        &
        Randomized
        &
        Reduction to set disjointness
        &
        Subsections~\ref{subsection:sum_integers}, \ref{sec:reduction-set-disjointness}
        and \cite{KS, Razborov, DBLP:journals/jcss/Bar-YossefJKS04} 		
        \\
        $O(\kappa) + 4n$
        &
        Shared randomness
        &
        Reduction to finding the intersection, set sizes are $\kappa$
        &
        Subsection~\ref{sec:finding_intersection} and \cite{BCKWY} 	
        \\	
        $O(\kappa) + 4n + O(\log n)$
        &
        Private randomness
        &
        Reduction to finding the intersection, set sizes are $\kappa$
        &
        Subsection~\ref{sec:finding_intersection} and \cite{BCKWY} 		
        \\
        $O( \kappa \cdot \log d_A + n)$
        &
        ``Las Vegas'' type
        &
        Set sizes are $\kappa$, $d_{A} = |S_{A} \setminus S_B|$
        &
        Section~\ref{sec:hash_functions}
        \\
        \hline
        \hline
    \end{tabular}
\end{center}
\caption{Communication complexity for a sum computation problem on the
reconciled data}
\label{table:complexity}
\end{table*}

\subsection{Upper Bound using Finding the Intersection Problem}
\label{sec:finding_intersection}

Another related problem is \emph{finding the intersection}~\cite{BCKWY}, in
which the users $A$ and $B$ are interested in finding the intersection of the sets that they possess. \smallskip

{\bf Finding the intersection problem:} there are two users $A$ and $B$ that possess the
sets $S_A, S_B \subseteq \FF^n$, respectively. The users want to compute jointly
the function $S_A \cap S_B$.\smallskip

A protocol for this problem can be used to
compute a sum (or, for example, a product) of the reconciled sets.

The following result is proved in~\cite{BCKWY} for the sets of size at most
$\kappa$.

\begin{theorem}~\cite[Theorem 3.1]{BCKWY}
There exists an $O(\sqrt{\kappa})$-round constructive randomized protocol for
finding the intersection problem with success probability $1 - 1/{\mbox{\sc
poly}(\kappa})$.  In the model of shared randomness the total communication
complexity is $O(\kappa)$ and in the model of private randomness it is $O(\kappa
+ \log n)$.
\label{thrm:intersection}
\end{theorem}

Assume that there is a protocol for computing the intersection $S_A \cap S_B$.
Then, the users can run the following protocol for computing the sum on the
reconciled data.
\begin{enumerate}
\item
$A$ and $B$ compute $S_A \cap S_B$.
\item
$A$ and $B$ compute $x_A = \sum_{x \in S_{A}} x$ and $x_B = \sum_{x \in S_{B}}
x$, respectively.
\item
$A$ and $B$ exchange the values of $x_A$ and $x_B$.
\item
Each user computes the result by computing $x_A + x_B -  \sum_{x \in S_{A} \cap
S_{B}} x$.
\end{enumerate}
By using Theorem~\ref{thrm:intersection}, the total number of communicated bits
is $O(\kappa) + 4n$ in the shared randomness model and  $O(\kappa) + 4n + O(\log
n)$ in the private randomness model.

\section{Using Hash Functions}
\label{sec:hash_functions}
\subsection{Setting}
\label{subsection:all:setting}

In this section, we construct a ``Las Vegas'' type randomized protocol for
computing the function $f$ as defined in~(\ref{eq:funct_sum}).

The proposed protocol is based on the use of universal hash functions~\cite{Carter-Wegman}, as follows. Let $H
\triangleq \FF^k$ and $\cH = \{ h \}$ be a family of all hash functions $h :
\FF^n \rightarrow H$, such that
\begin{equation}
\forall K \in H, \forall h \in \cH \; : \; | \left\{ x \; : \; h(x) = K \right\}
| = 2^{n-k} \; . \label{eq:hash_function}
\end{equation}
Assume that functions $h \in \cH$ are chosen randomly uniformly from $\cH$, and
independently from the previous choices.  Hereafter, we can assume that before
the protocol is executed, $A$ and $B$ agree on some random order of $h_0, h_1,
h_2, \cdots \in \cH$, which are used in the protocol.

\subsection{Protocol}

The pseudocode of the proposed protocol is presented as
Algorithm~\ref{alg:protocol}.

\begin{algorithm}[ht]
\begin{algorithmic}[1]
    \Procedure{Protocol}{}
        \For{$i=0$; true; $i=i+1$} \label{prot:loopstart}
            \State $B$ sends the set $K_i = \{h_i(x): x \in S_B\}$ to $A$
            \label{prot:send0}
            \State $A$ creates empty set $L_i$
            \For{$x \in S_A$}
                \If{$h_i(x) \not\in K_i$}
                    \State $A$ adds $x$ to $L_i$
                \EndIf
            \EndFor
            \If{$|L_i| = d_A$}
                \State \textbf{break} \label{prot:break}
            \EndIf
        \EndFor \label{prot:loopend}
        \State $A$ sends $s = \sum_{x \in L_i} x$ to $B$ \label{prot:send1}
        \State $B$ computes $s' = s + \sum_{x \in S_B} x$
        \State $B$ sends $s'$ to $A$ \label{prot:send2}
    \EndProcedure
\end{algorithmic}
\caption{Protocol pseudocode}
\label{alg:protocol}
\end{algorithm}

\subsection{Communication complexity}

Below, we estimate communication complexity of the proposed protocol.
While the main idea of the protocol is relatively straightforward,
the detailed analysis requires some nontrivial elaboration.

There are three statements, where the data is sent between the users: in lines
\ref{prot:send0}, \ref{prot:send1} and \ref{prot:send2}. We denote the
corresponding number of bits sent during each statement as $t_0$, $t_1$ and
$t_2$. We have:
\begin{align}
    t_0 &= k  m_B \; ,\label{eq:trans0} \\
    t_1 &= 2n-1 \; ,     \label{eq:trans1} \\
    t_2 &= 2n-1 \; .     \label{eq:trans2}
\end{align}

\subsection{Success Probability}
\label{subsection:all:success}

Below, we estimate the probability of the loop in lines
\ref{prot:loopstart}--\ref{prot:loopend} to end with a break statement in line
\ref{prot:break}. The number of loops determines the total number of
communicated bits.

In this analysis, we assume that the hash functions
satisfy~(\ref{eq:hash_function}).  Then, the collision probability for a
randomly chosen $h \in \cH$ is
\begin{eqnarray}
    \Pr[\mbox{collision}] & = &
		\Pr\left[h(x) = h(y) | x \in \FF^n, y \in \FF^n, x \neq y\right] \nonumber \\
		& = & \frac{2^{n-k}-1}{2^{n}-1} \; .
    \label{eq:coll}
\end{eqnarray}

The break statement in line \ref{prot:break} is activated when $|L_i| = d_A$ for
some $i$.

If $x \in S_0$, then $h(x) \in K_i$. Otherwise, if $x \in S_A \setminus S_0$,
then $h(x) \not\in K_i$ only if there is no collision with an element in $K_i$:
\begin{align}
    \Pr[|L_i|=d_A] & = \Pr[\mbox{no collision for every } x \in S_A \setminus S_0] \nonumber
    \\
    & = \left(1 - \frac{2^{n-k}-1}{2^{n}-1} \right)^{d_A} \; . \label{eq:stopping}
\end{align}

\subsection{Number of communicated bits}
\label{subsection:transbits}

Next, we compute the number of communicated bits $T_r$ during $r \in \nn$
rounds. For brevity, we denote
\begin{align}
    \pacc &= \Pr[\mbox{accept}] = \Pr[|L_i|=d_A]\\
    \pno &= \Pr[\mbox{not accept}] = 1 - \pacc.
\end{align}
Here, $\pacc$ is a probability that the protocol succeeds in computing the sum
of all elements.

At first, we look at the cases where we limit the number of rounds to 1, 2 and
3. To express the expected number of communicated bits in an instance of the
protocol, which succeeds after at most $r$ rounds, we use the random variable
$T_r$, $r \in \NN$. We have:
\begin{align*}
    E[T_1] & = \; \pacc (t_0 + t_1) + t_2 \; ,  \\
    E[T_2] & = \; \pacc (t_0 + t_1) + \pno \pacc (t_0 + t_0 + t_1) + t_2 \; , \\
    E[T_3] & = \; \pacc (t_0 + t_1) + \pno \pacc (t_0 + t_0 + t_1) \\
		& + \; \pno \pno  \pacc (t_0 + t_0 + t_0 + t_1) + t_2 \; .
\end{align*}

In general, when bounding the number of rounds by $r$, the number of the
communicated bits is
\begin{equation}
    E[T_r] = \sum_{i=0}^{r-1} \pno^i \pacc ( (i+1) t_0 + t_1 ) + t_2 \; .
    \label{eq:exp_transbits}
\end{equation}

By allowing an unbounded number of rounds, we obtain
\begin{align}
    E[T_\infty] - t_2 &= \sum_{i=0}^{\infty} \pno^i \pacc ( (i+1) t_0 + t_1) \nonumber \\
    & = \pacc t_0 \sum_{i=0}^\infty \pno^i (i+1) + \pacc t_1 \sum_{i=0}^\infty
    \pno^i \nonumber \\
    & = \pacc t_0 \frac{\pno}{(1-\pno)^2} + \pacc t_0 \frac{1}{1-\pno} + \pacc
    t_1 \frac{1}{1-\pno} \nonumber \\
    & = \pacc t_0 \frac{\pno}{\pacc^2} + \pacc t_0 \frac{1}{\pacc} + \pacc t_1
    \frac{1}{\pacc} \nonumber \\
    & = t_0 \frac{\pno}{\pacc} + t_0 + t_1 \nonumber \\
    & = t_0 \pacc^{-1} + t_1 \label{eq:exp_inf_transbits} \nonumber \\
    & = t_0 \left(1-\frac{2^{n-k}-1}{2^{n}-1}\right)^{-d_A} + t_1 \; .
\end{align}

By using equations~(\ref{eq:trans0})-(\ref{eq:trans2}), we obtain
\begin{equation}
    E[T_\infty] = k m_B \left(1-\frac{2^{n-k}-1}{2^{n}-1}\right)^{-d_A} + 4n - 2
    \; . \label{eq:transmitted_bits}
\end{equation}
Given $m_B$, $d_A$ and $n$, we next find
\[
    \argmin \limits_{k} k m_B \left(1-\frac{2^{n-k}-1}{2^{n}-1}\right)^{-d_A} +
    4n - 2\; ,
\]
in order to determine the optimal value of $\comm(F)$, which minimizes the total
number of communicated bits.

For simplicity, we assume that $k \ll n$ (otherwise, the hashing approach is not
efficient).  Under that assumption,
\[
    \comm(F) = \argmin \limits_{k} \left\{ k m_B (1-2^{-k})^{-d_A} + 4n - 2
\right\} \; .
\]
By substituting $k = \log_2 (\frac{d_A}{c})$, where $c$ is a constant, we
obtain:
\begin{align*}
    & k m_B (1-2^{-k})^{-d_A} + 4n - 2 \\
    & \approx k m_B \left(1- \frac{c}{d_A}\right)^{-d_A} + 4n-2 \\
    &= O( m_B \cdot \log d_A + n) \; .
\end{align*}

\section{Summary and Future Work}
\label{sec:summary}

In this work, we initiated a study of a new problem called \emph{function
computation on the reconciled data}. The problem considers a scenario where two
users possess sets of vectors $S_{A}$ and $S_{B}$, respectively, and they aim at
computing the value of $\phi(S_{A} \cup S_{B})$ for some function $\phi$. We
considered simple cases of $\phi$, such as identity, maximum, minimum, sum,
product. Specifically, for sum, we derived a number of lower and upper bounds on
communication complexity (for different models of randomness). We showed
connections to some known problems in communication complexity.  Finally, we
proposed a ``Las Vegas'' type randomized algorithm and analyzed its
communication complexity.

Many intriguing questions are still left open. Specifically, it would be
interesting to obtain tight bounds, and to design efficient protocols, for
computation of various functions. Different models of randomness can be
considered. Finally, protocols for a number of users larger than two can also be
investigated.

\section{Acknowledgements}
The authors wish to thank Dirk Oliver Theis for helfpul discussions and for
pointing out the connection to the set disjointness problem.

\end{document}